\documentclass[conference]{IEEEtran}
\usepackage{cite}

\ifCLASSINFOpdf
 \usepackage[pdftex]{graphicx}
 \usepackage{subfigure}
\else
\fi
\usepackage{amsmath}
\usepackage{amssymb}
\usepackage{amsthm}

\usepackage{float}

\newtheorem{lemma}{\textbf{Lemma}}

\allowdisplaybreaks
\begin{document}
%
\title{Age-of-Information in the Presence of Error}

\author{\IEEEauthorblockN{Kun Chen}
\IEEEauthorblockA{IIIS, Tsinghua University\\
chenkun14@mails.tsinghua.edu.cn}
\and
\IEEEauthorblockN{Longbo Huang}
\IEEEauthorblockA{IIIS, Tsinghua University\\
longbohuang@tsinghua.edu.cn}
 
}


%


\maketitle

\begin{abstract}
We consider the  peak age-of-information (PAoI) in an $M/M/1$ queueing  system with packet delivery error, i.e., update packets can get lost during transmissions to their destination. We focus on two types of policies, one is to adopt Last-Come-First-Served (LCFS) scheduling, and the other is to utilize retransmissions, i.e., keep transmitting the most recent packet. Both policies can effectively avoid the queueing delay of a busy channel and ensure a small PAoI. 
Exact PAoI expressions under both  policies with different error probabilities are derived, including First-Come-First-Served (FCFS), LCFS with preemptive priority, LCFS with non-preemptive priority, Retransmission with preemptive priority, and Retransmission with non-preemptive priority. Numerical results obtained from analysis and simulation are presented to validate our results.
\end{abstract}


%
\IEEEpeerreviewmaketitle

\section{Introduction} 
Many information systems work in such a mode that status updates are first collected from a time-varying environment, and then control decisions are made based on these information. 
Examples include sensor networks for large-scale monitoring \cite{corke}, vehicular networks where vehicle  position and velocity information are disseminated to assist safe and intelligent transportation \cite{papa}, and  wireless networks where scheduling is carried out based on channel state information \cite{reddy}. A key to these systems is to \emph{ensure timely delivery of  status updates}, since out-of-date information can lead to incorrect system status estimation and result in severe performance loss. 

\emph{Age-of-information} (AoI), first proposed in \cite{kaul}, provides a measure for the ``freshness" of the current status information, and is an important metric for measuring quality-of-service (QoS) of a system. 
Different from typical performance metrics such as delay or throughput, AoI jointly captures the latency in transmitting updates and the rate at which they are delivered. 

There have been various recent works on understanding AoI. 
\cite{kaul} analyzes AoI for queueing models including $M/M/1$, $M/D/1$ and $D/M/1$. A more complicated case with multiple update sources is analyzed in \cite{yates}.
\cite{kaul2} studies AoI in a Last-Come-First-Served (LCFS) $M/M/1$ queueing system with or without preemption.
The case when the destination may receive out-of-order packets is considered in \cite{kam}. 
In \cite{costa}, the authors introduce a notion \emph{peak age-of-information} (PAoI) and consider systems with packet management, i.e., the queue can choose to only keep a subset of update packets.  
%
AoI in a multi-class $M/G/1$ queueing  system  is studied in \cite{huang}. In \cite{guo}, the authors study optimal update scheduling in a discrete-time multi-source system. The optimal update generating policy is explored in \cite{sun}. 

We notice that one common assumption made in most aforementioned works is that update packet delivery is always perfect, and AoI has been investigated mostly under the First-Come-First-Served (FCFS) principle. An exception is \cite{kaul2}, which studies AoI under the LCFS principle, but also assumes perfect packet delivery.
However, in practical systems, packet transmissions often contain errors and losses, e.g., due to interference or buffer overflow at intermediate routers in a multi-hop network.
%
To study the impact of such delivery errors on AoI,  in this paper, we focus on an $M/M/1$ queueing model where each packet, upon service completion, arrives at the destination with a nonzero probability. Our model captures (i) the queueing effect, which approximates the process where update packets are sent over a channel or a network and can cause congestion (This is different from \cite{guo}, which also considers transmission errors), and (ii) the  error component, which models the fact that update packets can get lost during the delivery process.  

We first focus on the LCFS service principle and derive the exact PAoI for both the systems with preemptive priority and non-preemptive priority. Intuitively, LCFS is good for two reasons. (i) Compared to packet management schemes, e.g, \cite{costa}, LCFS similarly avoids delaying new update packets with queueing by letting them go first. This results in significant reduction of AoI compared to FCFS, especially when the channel utilization is high. (ii) When there are errors in packet transmissions, packet management schemes can suffer severely due to the lack of updates to deliver, while LCFS still ensures a good delivery rate and does not affect AoI significantly.
%

Next we analyze  the PAoI under retransmission schemes. Here we do not assume feedback, since retransmissions based on feedback may suffer from waste of time waiting for feedback, or interference between update packets and feedback information. Thus, the Retransmission policies refer to keep transmitting the most recent packet repeatedly until a new packet arrives. 
Compared to LCFS, retransmission policies have an advantage of 
always transmitting the most recent updates,  at the cost of additional packet state management.  We also derive the exact PAoI expressions for retransmission with or without preemption. 

In this work LCFS and Retransmission policies are both studied to cover various scenarios. Although utilizing retransmissions is expected to contribute to a small AoI, it does not apply to scenarios where transmissions are not guaranteed, e.g., UDP and some wireless sensor networks. 
The rest of the paper is organized as follows. In Section \ref{section:model} we introduce the model. 
In Sections \ref{section:fcfs}, \ref{section:preemptive}, \ref{section:nonpreemptive}, \ref{section:retran-preemptive} and \ref{section:retran-nonpreemptive} we present our analysis for the FCFS, LCFS preemptive, LCFS non-preemptive, Retransmission preemptive and Retransmission non-preemptive cases. In Section \ref{section:numerical} we present numerical results. We conclude the paper in Section \ref{section:conclusion}.

\section{System Model}\label{section:model}
We consider a system where  a source transmits updates (packets) to a remote destination through a queue. The source generates packets according to a Poisson process with rate $\lambda$.  The service time for each packet is exponentially distributed with service rate $\mu$. 

Different from previous works, we assume that upon service completion, each packet arrives at the destination 
independently with probability $p\in[0, 1]$. Such a system is modeled by an $M/M/1$ queueing system with packet loss, as shown in Fig. \ref{fig_LCFS}. 
The packet loss model captures real-world situations where update packets can get lost during delivery to their destination, e.g., interference or buffer overflow, and has not yet been studied. 
\begin{figure}[ht] 
\centering
\vspace{-.1in}
\includegraphics[width=2.0in]{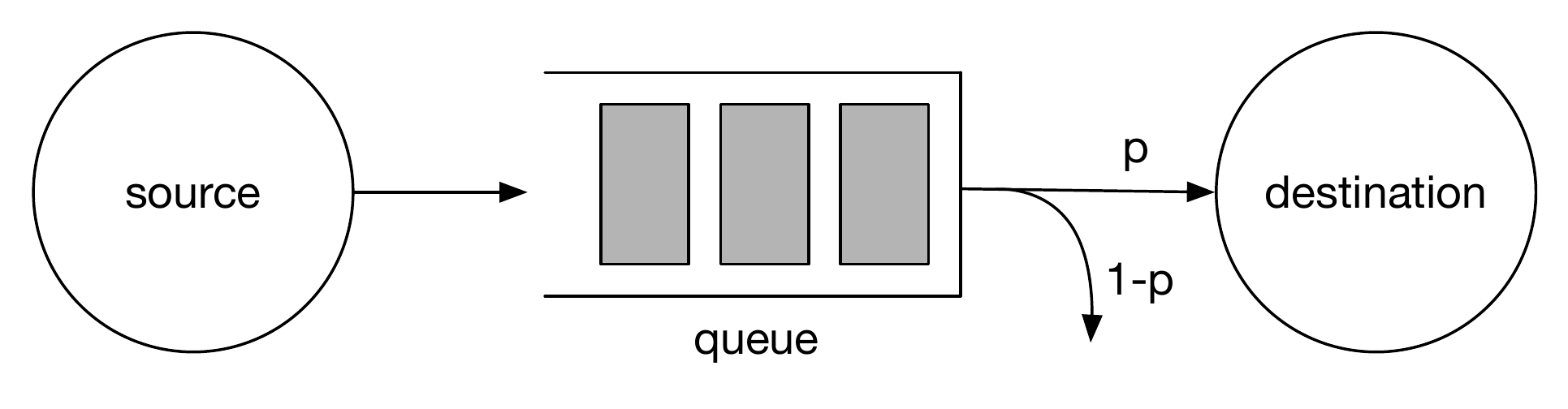}
\vspace{-.1in}
\caption{The packet delivery process in a queueing  system with packet loss.}
\label{fig_LCFS}
\end{figure}

We study the peak age-of-information (PAoI) in this system, which is defined as follows. 
Suppose each update packet has a time-stamp, marking its generation time. Denote the  time-stamp of the most recently received update at time $t$ as $\delta(t)$. Then,  the status age is defined as \cite{kaul}: 
$$\Delta(t)\triangleq t-\delta(t),$$ 
and the set of peak ages is defined as:   
$$\{\Delta(t_{i})|\exists \epsilon>0 \text{ s.t. } \forall t\in (t_{i}-\epsilon)\cup(t_{i}+\epsilon), \Delta(t)<\Delta(t_{i})\}.$$
Then, PAoI \cite{costa} is defined to be:   
$$A_P\triangleq \lim_{I\to\infty}\frac{1}{I}\sum_{i=1}^I\Delta_i=\mathbb{E}\{\Delta_i\},$$ 
where $\Delta_i=\Delta(t_i)$ is the $i$-th peak of $\Delta(t)$ (See Fig. \ref{fig:paoi}). The last equality follows from the ergodicity of $\Delta_i$. As shown in \cite{huang}, PAoI is closely related to the average AoI, but  is much more tractable. 

We first introduce some useful definitions. 
Denote $N$ the set of all packets, according to the order in which they arrive. 
For a packet $n$, denote $a(n)$ its arrival time, $d(n)$  its departure time, and $u(n)$ the time it starts to receive service. 
Let $\Phi$ denote the set of all successfully transmitted packets. Under the LCFS service discipline, a successfully transmitted packet may be outdated when arriving at the destination. Thus,  we further define the set of \emph{informative} packets  $\Psi$ as: 
\begin{eqnarray*}
\Psi \triangleq \{n\in \Phi \,|\,d(n)-a(n)<\Delta(d(n))\}. 
\end{eqnarray*} 
That is,  $\Psi=\{n_1,n_2,\dots,n_i,\dots\}$ contains the packets which offer new information (so the system age decreases) when they reach the destination.

Regarding the evolution of the system, we define the following random variables:
\begin{align*}
X_n&\triangleq a(n+1)-a(n),\\
W_n&\triangleq u(n)-a(n), \\
S_n&\triangleq d(n)-u(n),
\end{align*}
i.e., $X_{n}$ is the inter-arrival time between $n$ and $n+1$; $W_{n}$ is the waiting time of $n$; $S_{n}$ is  the ``service time'' of $n$. Note that in the LCFS with preemptive priority case, $S_{n}$ may include service time of later packets if $n$ is preempted by other packets.
%

\section{PAoI under FCFS}\label{section:fcfs}

For the basic First-Come-First-Served (FCFS) case with packet loss, PAoI can be easily obtained. 
Define the first informative packet which arrives no earlier that $n$ as 
\begin{eqnarray*}
\alpha(n)\triangleq\min\{n_i|n_i\in\Psi,\ a(n_i)\ge a(n)\}.
\end{eqnarray*}
Moreover, define the inter-arrival time between $n$ and $\alpha(n)$ as
\begin{eqnarray*}
\hat{X}_n\triangleq a(\alpha(n))-a(n),
\end{eqnarray*}
and define $\hat{S}_n$ as the time duration from the moment $n$ starts to receive service to the moment $\alpha(n)$ departs, i.e., 
\begin{eqnarray*}
\hat{S}_n\triangleq d(\alpha(n))-u(n).
\end{eqnarray*}
Note that if $n\in\Psi$, we have $\alpha(n)=n$, $\hat{X}_n=0$ and $\hat{S}_n=S_{n}$.
Since $\Phi=\Psi$ under FCFS, PAoI is composed of the (expected) inter-arrival time of two successfully transmitted packets, plus the time a packet spends in the system. Thus,
\begin{align*}
A_P^{FCFS}&=\mathbb{E}\{X_{n_i}+\hat{X}_{n_i+1}+W_{n_{i+1}}+S_{n_{i+1}}|n_i,n_{i+1} \in \Phi\},\\
&=\frac{1}{\lambda}+\mathbb{E}\{\hat{X}_{n_i+1}\}+\frac{1}{\mu-\lambda},
\end{align*}
where $n_i+1$ is $n_i$'s next packet and $n_{i+1}$ is $n_i$'s next packet in $\Phi$. Moreover, 
\begin{eqnarray*}
\mathbb{E}\{\hat{X}_{n_i+1}\}&=&p\mathbb{E}\{\hat{X}_{n_i+1}|n_i+1\in\Phi\}+\\
&&\quad(1-p)\mathbb{E}\{X_{n_i+1}+\hat{X}_{n_i+2}|n_i+1\notin\Phi\}\\
&=&0+(1-p)(\frac{1}{\lambda}+\mathbb{E}\{\hat{X}_{n_i+1}\}),
\end{eqnarray*}
where we have used $\mathbb{E}\{\hat{X}_{n_i+1}\}=\mathbb{E}\{\hat{X}_{n_i+2}\}$. Thus, 
\begin{eqnarray*}
\mathbb{E}\{\hat{X}_{n_i+1}\}&=&\frac{1-p}{p\lambda},
\end{eqnarray*}
which implies: 
\begin{eqnarray}
A_P^{FCFS}=\frac{1}{p\lambda}+\frac{1}{\mu-\lambda}. \label{paoi_fcfs}
\end{eqnarray}
However, the FCFS policy, as discussed above, can suffer from   traffic congestion, under which each packet will take a long time to get through the queue and the PAoI can be poor. 
Thus, in this work, we focus on the  Last-Come-First-Served (LCFS) as well as Retransmission policies and consider the  following two scheduling schemes. 
\begin{enumerate}
\item  Preemptive priority: If a new packet arrives while the server is busy, it preempts the current packet and starts service immediately.
\item Non-preemptive priority: The server always completes the current packet and then starts serving the most recent packet in the queue. 
\end{enumerate}
The reasons to focus on PAoI with LCFS  are as follows: (i) Intuitively, letting later packets go earlier should make the status at the destination fresher. Hence, the PAoI will be much smaller. (ii) Compared to packet management schemes where packets are often dropped for queue size reduction, e.g.,  \cite{costa}, LCFS still transmits all packets. Thus, in the case when errors can occur and packets can get lost, LCFS ensures that the destination still gets updates more regularly, maintaining a lower level of PAoI. 
Note that characterizing PAoI under LCFS, even with perfect delivery, is nontrivial and has not been studied before, especially for the non-preemptive case. 

We also analyze the PAoI under Retransmission policies, which have an advantage over LCFS, as they always deliver the most latest information. On the other hand, retransmission requires additional packet state management. 


\section{PAoI under LCFS with Preemptive Priority}\label{section:preemptive}


\begin{figure}
  \centering
  \subfigure[Preemptive]{
    \label{fig_preem} 
    \includegraphics[width=1.65in]{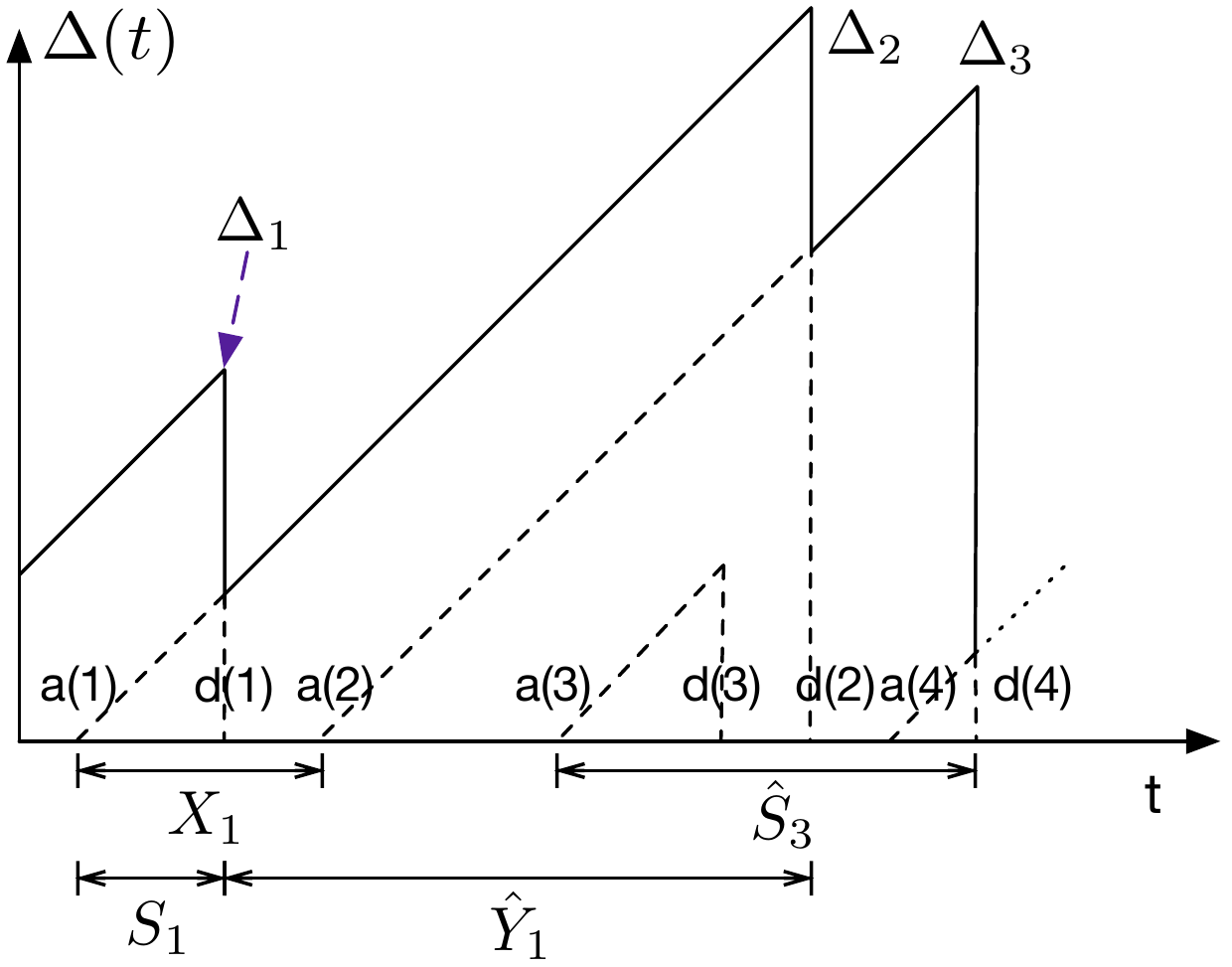}}
  \subfigure[Non-preemptive]{
    \label{fig_nonpre} 
    \includegraphics[width=1.65in]{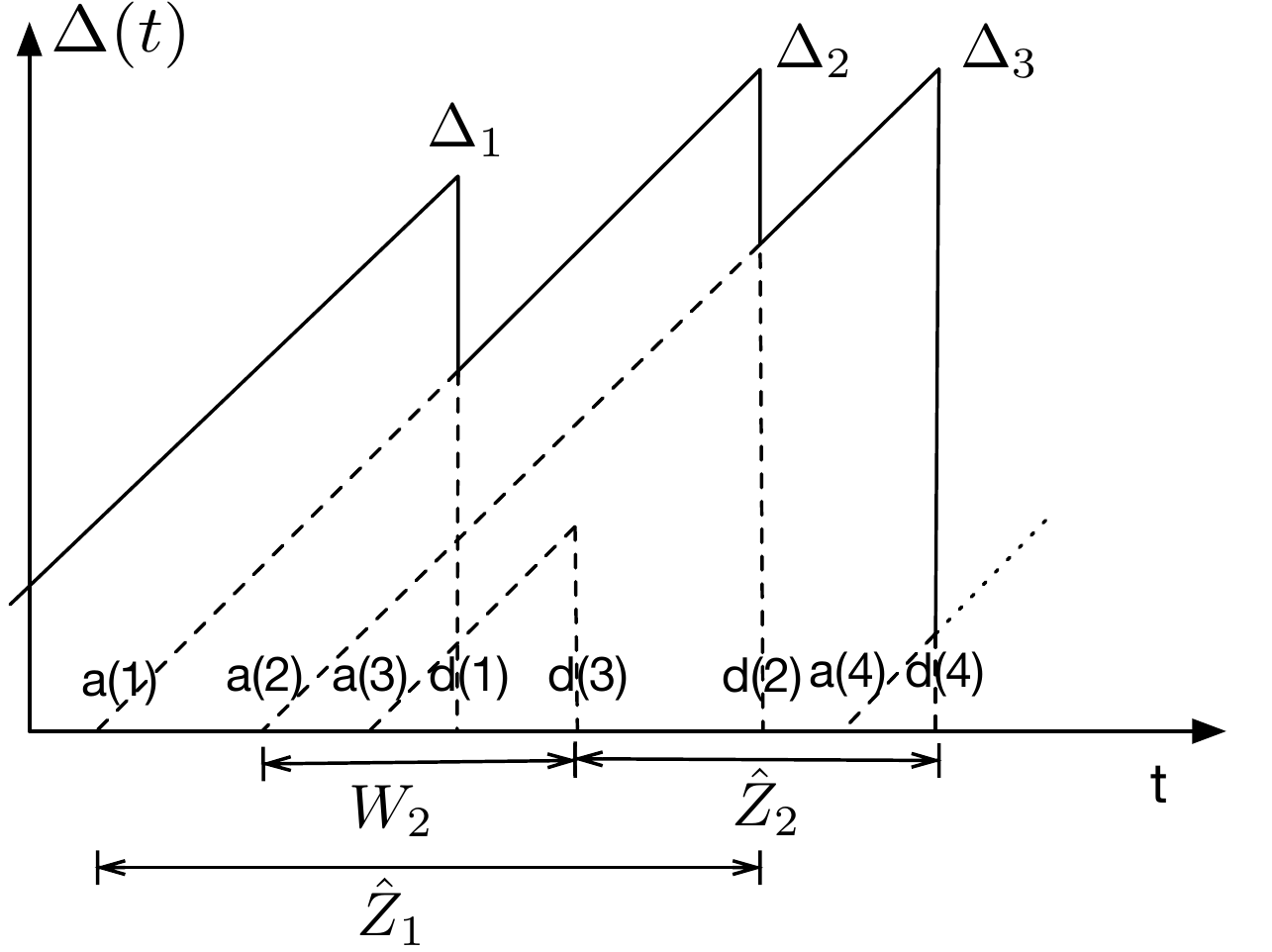}}
  \caption{Evolution of status age in the LCFS $M/M/1$ system. PAoI is divided in different ways under the preemptive and non-preemptive cases.}
  \label{fig:paoi}
\vspace{-.2in}
\end{figure}
We begin with LCFS with preemptive priority. Note that in this case, $u(n)=a(n),\forall n$, i.e., packets get served immediately upon arrival. 
%
Moreover, in this case $\{\hat{S}_n\}_n$ are statistically the same. 
As shown in Fig. \ref{fig_preem}, PAoI is the elapsed time from the moment when a packet $n_i\in \Psi$ arrives,  until the moment when  $n_{i+1}\in \Psi$ departs (recall that $\Psi$ denotes the set of informative packets). Define the first informative packet which arrives after $n$'s departure as 
\begin{eqnarray*}
\beta(n)\triangleq\min\{n_i|n_i\in\Psi,\ a(n_i)> d(n)\},
\end{eqnarray*}
and the inter-departure time between $n$ and $\beta(n)$:
\begin{eqnarray*}
\hat{Y}_n\triangleq d(\beta(n))-d(n).
\end{eqnarray*}
Since the packets arriving after $a(n_i)$ but before $d(n_i)$ preempt $n_i$ and get lost upon departure (because $n_i\in\Psi$), we have (see Fig. \ref{fig_preem}):
\begin{eqnarray}
A_P^{LCFS,pre}=\mathbb{E}\{S_{n_i}+\hat{Y}_{n_i}|n_i\in\Psi\}.
\end{eqnarray}

\subsection{Analyzing a Service Process}
Here we use $S_n$ to also denote the process of serving a packet $n$. For simplicity,  we define the following symbols (notice that in other sections these symbols may have different definitions): 
\begin{eqnarray*} 
\tilde{p}&\triangleq& \mathbb{P}(\hat{S}_n\le S_n), \\
\tilde{t}&\triangleq& \mathbb{E}\{\hat{S}_n|\hat{S}_n\le S_n\}, \\
\tilde{s}&\triangleq& \mathbb{E}\{S_n|\hat{S}_n > S_n\},
\end{eqnarray*}
i.e., $\tilde{p}$ is the probability that there exists a packet that reaches the destination successfully during $S_n$ (including $n$ and the packets arriving after $a(n)$ but before $d(n)$) .

%
We first have the following lemma, based on which we will derive $\tilde{t}$ and $\tilde{s}$.
\begin{lemma}\label{lem}
For a nonnegative random variable $X$, an event $E$ and a sequence of events $E_1,E_2,\dots,E_K$ which satisfies $E_i\cap E_j=\emptyset ,\forall i\ne j$ and $E=\cup_{k=1}^K E_k$, we have
$$\mathbb{P}(E)\mathbb{E}\{X|E\}=\sum_{k=1}^K\mathbb{P}(E_k)\mathbb{E}\{X|E_k\}. $$
\end{lemma}
\begin{proof}
\begin{eqnarray*}
\mathbb{P}(E)\mathbb{E}\{X|E\}&=&\mathbb{P}(E)\int_0^\infty\mathbb{P}(X>x|E)dx \\
&=&\mathbb{P}(E)\int_0^\infty\frac{\mathbb{P}(X>x,E)}{\mathbb{P}(E)}dx \\
&=&\int_0^\infty\sum_{k=1}^K\mathbb{P}(X>x,E_k)dx \\
&=&\sum_{k=1}^K\mathbb{P}(E_k)\int_0^\infty\frac{\mathbb{P}(X>x,E_k)}{\mathbb{P}(E_k)}dx \\
&=&\sum_{k=1}^K\mathbb{P}(E_k)\mathbb{E}\{X|E_k\}
\end{eqnarray*}
\vspace{-.1in}
\end{proof}
The probability that $X_n\le S_n$ is $\frac{\lambda}{\lambda+\mu}$. If that happens, the system will first serve packet $n+1$ (during which other new packets may come and complete service before $n+1$), then continue the service of $n$. 
Based on this observation, we have: 
\begin{align}
\tilde{p}=&\frac{\mu}{\lambda+\mu}p+\frac{\lambda}{\lambda+\mu}[\mathbb{P}(\hat{S}_{n+1}\le S_{n+1}) +\mathbb{P}(\hat{S}_{n+1}> S_{n+1})\nonumber\\
&\qquad\times\mathbb{P}(\hat{S}_n\le S_n|X_n\le S_n,\hat{S}_{n+1}> S_{n+1})]\nonumber\\
=&\frac{\mu}{\lambda+\mu}p+\frac{\lambda}{\lambda+\mu}[\tilde{p}+(1-\tilde{p})\tilde{p}].  \label{preem_pt_eq1}
\end{align}
Furthermore, we have: 
\begin{align}
\tilde{p}\tilde{t}=&\frac{\mu}{\lambda+\mu}p\mathbb{E}\{S_n|X_n>S_n,n\in\Phi\}+ \frac{\lambda}{\lambda+\mu}\big[\tilde{p}\mathbb{E}\{X_n+\nonumber\\
&\qquad\hat{S}_{n+1}|X_n\le S_n,\hat{S}_{n+1}\le S_{n+1}\}+(1-\tilde{p})\tilde{p} \nonumber \\
&\qquad\times\mathbb{E}\{\hat{S}_n|X_n\le S_n,\hat{S}_{n+1}> S_{n+1},\hat{S}_{n}\le S_{n}\}\big] \nonumber \\
=&\frac{\mu}{\lambda+\mu}p\frac{1}{\lambda+\mu}+ \frac{\lambda}{\lambda+\mu}\big[\tilde{p}(\frac{1}{\lambda+\mu}+\tilde{t})\nonumber \\
&\qquad+(1-\tilde{p})\tilde{p}(\frac{1}{\lambda+\mu}+\tilde{s}+\tilde{t})\big], \label{preem_tt_eq}
\end{align}
\begin{align}
\hspace{-.05in}(1-\tilde{p})\tilde{s}=&
\frac{\mu(1-p)}{\lambda+\mu}\mathbb{E}\{S_n|X_n>S_n,n\notin\Phi\}+\frac{\lambda}{\lambda+\mu}(1-\tilde{p})^2 \nonumber\\
&\qquad\times\mathbb{E}\{S_n|X_n\le S_n,\hat{S}_{n+1}> S_{n+1},\hat{S}_{n}> S_{n}\} \nonumber\\
=&\frac{\mu(1-p)}{\lambda+\mu}\frac{1}{\lambda+\mu}+\frac{\lambda}{\lambda+\mu}(1-\tilde{p})^2(\frac{1}{\lambda+\mu}+2\tilde{s}).
\label{preem_st_eq}
\end{align}
In the above, we have used
\begin{eqnarray*}
\hspace{-.4in}&&\mathbb{P}(\hat{S}_n\le S_n|X_n\le S_n,\hat{S}_{n+1}> S_{n+1}) \\
\hspace{-.4in}&&=\mathbb{P}(\hat{S}_n-X_{n}-S_{n+1}\le S_n-X_{n}-S_{n+1}|\\
\hspace{-.4in}&&\qquad\quad \hat{S}_{n}> X_{n}+S_{n+1}) \\
\hspace{-.4in}&&=\tilde{p},
\end{eqnarray*}
 and that 
\begin{eqnarray*}
\hspace{-.3in}&&\mathbb{E}\{\hat{S}_n|X_n\le S_n,\hat{S}_{n+1}> S_{n+1},\hat{S}_{n}\le S_{n}\}\\
\hspace{-.3in}&&=\mathbb{E}\{X_n+S_{n+1}+\hat{S}_n-X_{n}-S_{n+1}|\\
\hspace{-.3in}&&\qquad\quad X_n\le S_n,\hat{S}_{n+1}> S_{n+1},\hat{S}_{n}\le S_{n}\} \\
\hspace{-.3in}&&=\mathbb{E}\{X_n|X_n\le S_n\}+\mathbb{E}\{S_{n+1}|\hat{S}_{n+1}> S_{n+1}\}\\
\hspace{-.3in}&&\qquad\quad +\mathbb{E}\{\hat{S}_n-X_{n}-S_{n+1}|\hat{S}_{n}> X_{n}+S_{n+1},\hat{S}_{n}\le S_{n}\} \\
\hspace{-.3in}&&=\frac{1}{\lambda+\mu}+\tilde{s}+\tilde{t},
\end{eqnarray*}
since both the services and arrivals are memoryless. We  get from \eqref{preem_pt_eq1} that: 
\begin{eqnarray}
\lambda\tilde{p}^2+(\mu-\lambda)\tilde{p}-\mu p=0,
\label{preem_pt_eq2}
\end{eqnarray}
which leads to: 
\begin{eqnarray}
\tilde{p}=\frac{-(\mu-\lambda)+\sqrt{(\mu-\lambda)^2+4\lambda\mu p}}{2\lambda}. 
\label{preem_pt_res}
\end{eqnarray}
Solving \eqref{preem_tt_eq} and \eqref{preem_st_eq}, and using  \eqref{preem_pt_eq2} give us: 
\begin{eqnarray}
(1-\tilde{p})\tilde{s}&=&\frac{1-\tilde{p}}{\mu-\lambda+2\lambda\tilde{p}},
\label{preem_st_res}\\
\tilde{p}\tilde{t}&=&\frac{\tilde{p}+\frac{\mu}{\mu-\lambda+2\lambda\tilde{p}}(\tilde{p}-p)}{\mu-\lambda+\lambda\tilde{p}}. 
\label{preem_tt_res}
\end{eqnarray}

On the other hand, $n\in\Psi$ means that only $n$ reaches the destination successfully during $S_n$, which is equivalent to $\hat{S}_n=S_n$. Therefore, we get: 
\begin{eqnarray*}
\mathbb{P}(n\in\Psi)&=&\frac{\mu}{\lambda+\mu}p+\frac{\lambda}{\lambda+\mu}(1-\tilde{p})\\
&&\quad\times\mathbb{P}(\hat{S}_n-X_n-S_{n+1}= S_n-X_n-S_{n+1}),
\end{eqnarray*}
\begin{eqnarray*}
\mathbb{P}(n\in\Psi)\mathbb{E}\{S_n|n\in\Psi\}&=&\frac{\mu}{\lambda+\mu}p\frac{1}{\lambda+\mu}+\frac{\lambda}{\lambda+\mu}(1-\tilde{p}) \nonumber\\
&&\quad\times\mathbb{P}(\hat{S}_n= S_n)(\frac{1}{\lambda+\mu}+\tilde{s} \nonumber\\
&&\quad+\mathbb{E}\{S_n|n\in\Psi\}),
\end{eqnarray*}
where we have used $\mathbb{E}\{S_n-X_n-S_{n+1}|S_n>X_n+S_{n+1}, n\in\Psi\}=\mathbb{E}\{S_n|n\in\Psi\}$. As a result, 
\begin{eqnarray*}
\mathbb{P}(n\in\Psi)&=&\frac{\mu p}{\mu+\lambda\tilde{p}}\\
\mathbb{E}\{S_n|n\in\Psi\}&=&\frac{1}{\mu-\lambda+2\lambda\tilde{p}}.
\end{eqnarray*}

\vspace{-.1in}
\subsection{Computing PAoI}
\vspace{-.05in}


Now consider $\mathbb{E}\{\hat{Y}_{n_i}|n_i\in\Psi\}=\mathbb{E}\{\hat{Y}_{n_i}\}$. Suppose the first packet which arrives after $d(n_i)$ is $\tilde{n}_i$. Since the exponential distribution is memoryless, the expected time from $d(n_i)$ to $a(\tilde{n}_i)$ is $\frac{1}{\lambda}$. If $\hat{S}_{\tilde{n}_i}\le S_{\tilde{n}_i}$, the (expected) remaining  time of $\hat{Y}_{n_i}$ from $a(\tilde{n}_i)$ is $\tilde{t}$. Otherwise the remaining time is $\tilde{s}+\mathbb{E}\{\hat{Y}_{\tilde{n}_i}\}$. Based on the above analysis, 
\begin{eqnarray*}
\mathbb{E}\{\hat{Y}_{n_i}\}=\frac{1}{\lambda}+\tilde{p}\tilde{t}+(1-\tilde{p})(\tilde{s}+\mathbb{E}\{\hat{Y}_{\tilde{n}_i}\}),
\end{eqnarray*}
from which we obtain: 
\begin{eqnarray*}
\mathbb{E}\{\hat{Y}_{n_i}\}=\frac{1}{\tilde{p}}\big[\frac{1}{\lambda}+\tilde{p}\tilde{t}+(1-\tilde{p})\tilde{s}\big]. 
\end{eqnarray*}
Substituting \eqref{preem_st_res} and \eqref{preem_tt_res} into the above gives us: 
\begin{eqnarray*}
\mathbb{E}\{\hat{Y}_{n_i}\}=\frac{\mu(\mu-\lambda)+2\lambda\mu p+\lambda(\lambda+\mu)\tilde{p}}{\lambda\mu p(\mu-\lambda+2\lambda\tilde{p})}.
\end{eqnarray*}
As a result, 
\begin{align}
A_P^{LCFS,pre}=&\mathbb{E}\{S_{n_i}|n_i\in\Psi\}+\mathbb{E}\{\hat{Y}_{n_i}\} \nonumber\\
=&\frac{1}{\mu-\lambda+2\lambda\tilde{p}}+\frac{\mu(\mu-\lambda)+2\lambda\mu  p+\lambda(\lambda+\mu)\tilde{p}}{\lambda\mu p(\mu-\lambda+2\lambda\tilde{p})} \nonumber\\
=&\frac{\mu(\mu-\lambda)+3\lambda\mu p+\lambda(\lambda+\mu)\tilde{p}}{\lambda\mu p(\mu-\lambda+2\lambda\tilde{p})}, \label{paoi_preem}
\end{align}
where $\tilde{p}$ is given in \eqref{preem_pt_res}.
In the case when $p=1$, the above result becomes $PAoI=\frac{1}{\lambda+\mu} + \frac{1}{\lambda}+\frac{1}{\mu}$.

\section{PAoI  under LCFS with Non-Preemptive Priority}\label{section:nonpreemptive}
In this case, if a new packet arrives while the server is busy, it cannot interrupt the current service. 
%
From Fig. \ref{fig_nonpre}, we see that  PAoI is similarly the elapsed time from the moment when a packet $n_i\in \Psi$ arrives, to the moment when $n_{i+1}$ departs.  
%
Define the first informative packet which arrives after $n$ starts to receive service as 
\begin{eqnarray*}
\gamma(n)\triangleq\min\{n_i|n_i\in\Psi,\ a(n_i)> u(n)\}.
\end{eqnarray*}
and the time duration from the moment $n$ starts to receive service to the moment $\gamma(n)$ departs as
\begin{eqnarray*}
\hat{Z}_n\triangleq d(\gamma(n))-u(n).
\end{eqnarray*}
Since the packets arriving after $a(n_i)$ but before $u(n_i)$ are served before $n_i$ and get lost upon departure (because $n_i\in\Psi$), we have (see Fig. \ref{fig_nonpre}):
\begin{eqnarray}
A_P^{LCFS,non}=\mathbb{E}\{W_{n_i}+\hat{Z}_{n_i}|n_i\in\Psi\}.
\end{eqnarray}

\vspace{-.05in}
\subsection{Analyzing a Service Process}
We first define $\bar{S}_n$ as the process since $u(n)$ till the first time the server becomes free or starts to serve a packet that arrives no later than $u(n)$ (excluding $n$). Since $\bar{S}_n$ is determined by the services and arrivals after $u(n)$ and independent of the system state at $u(n)$ and the history before $u(n)$, the $\bar{S}_n$ processes  induced by different packets $n$ are identically distributed. 
We re-define the following symbols: 
\begin{eqnarray*}
\tilde{p}&\triangleq& \mathbb{P}(\hat{Z}_n\le \bar{S}_n), \\
\tilde{t}&\triangleq& \mathbb{E}\{\hat{Z}_n|\hat{Z}_n\le \bar{S}_n\}, \\
\tilde{s}&\triangleq& \mathbb{E}\{\bar{S}_n|\hat{Z}_n > \bar{S}_n\},
\end{eqnarray*}
i.e., $\tilde{p}$ is the probability that there exists a packet which arrives after $u(n)$ and reaches the destination successfully during $\bar{S}_n$.

Consider $\bar{S}_n$. 
Suppose the number of packets arriving during $S_n$ is $\sigma(S_n)$. We have $\forall k\ge 0$,
\begin{eqnarray*}
p(\sigma(S_n)=k)&=&(\frac{\lambda}{\lambda+\mu})^k\frac{\mu}{\lambda+\mu}, \\
\mathbb{E}\{S_n|\sigma(S_n)=k\}&=&\frac{k+1}{\lambda+\mu}.
\end{eqnarray*}
If $\sigma(S_n)=k>0$ (which is needed for $\hat{Z}_n\le \bar{S}_n$), when $n$ completes service, the system will serve the $(n+k)$-th packet and enter $\bar{S}_{n+k}$. 
If $n+k\in\Phi$,  then $\hat{Z}_n\le \bar{S}_n$ and the remaining time of $\hat{Z}_n$ from $d(n)$ is $\mathbb{E}\{S_{n+k}|\sigma(S_n)=k,n+k\in\Phi\}=\frac{1}{\mu}$. 
If $n+k\notin\Phi$, then if $\hat{Z}_{n+k}\le \bar{S}_{n+k}$,  $\hat{Z}_n\le \bar{S}_n$ and the remaining time of $\hat{Z}_n$ from $d(n)$ is $\mathbb{E}\{\hat{Z}_{n+k}|\sigma(S_n)=k,n+k\notin\Phi, \hat{Z}_{n+k}\le \bar{S}_{n+k}\}=\mathbb{E}\{\hat{Z}_{n+k}|\hat{Z}_{n+k}\le \bar{S}_{n+k}\}=\tilde{t}$. Similar analysis applies to the $(n+k-1)$-th, the $(n+k-2)$-th, $\cdots$, and the $(n+1)$-th packet. Thus, using Lemma \ref{lem},
\begin{eqnarray}
\hspace{-.3in}\tilde{p}&=&\sum_{k=1}^\infty\frac{\mu}{\lambda+\mu}(\frac{\lambda}{\lambda+\mu})^k\big\{p+(1-p)\tilde{p} \nonumber\\
\hspace{-.3in}&&\qquad +(1-p)(1-\tilde{p})[p+(1-p)\tilde{p}]+\cdots \nonumber\\
\hspace{-.3in}&&\qquad+(1-p)^{k-1}(1-\tilde{p})^{k-1}[p+(1-p)\tilde{p}]\big\}\nonumber\\
\hspace{-.3in}&=&\frac{\lambda}{\lambda+\mu}-\frac{\mu}{\lambda+\mu}\frac{\frac{\lambda}{\lambda+\mu}(1-p)(1-\tilde{p})}{1-\frac{\lambda}{\lambda+\mu}(1-p)(1-\tilde{p})},  \label{nonpre_pt_eq1} 
\end{eqnarray}
\begin{eqnarray}
\hspace{-.3in}\tilde{p}\tilde{t}&=&\sum_{k=1}^\infty\frac{\mu}{\lambda+\mu}(\frac{\lambda}{\lambda+\mu})^k\bigg\{p(\frac{k+1}{\lambda+\mu}+\frac{1}{\mu}) \nonumber\\
\hspace{-.3in}&&+(1-p)\tilde{p}(\frac{k+1}{\lambda+\mu}+\tilde{t})+(1-p)(1-\tilde{p}) \nonumber\\
\hspace{-.3in}&& \qquad\times\big[p(\frac{k+1}{\lambda+\mu}+\tilde{s}+\frac{1}{\mu})+(1-p)\tilde{p}(\frac{k+1}{\lambda+\mu} \nonumber\\
\hspace{-.3in}&& +\tilde{s}+\tilde{t})\big]+\cdots+(1-p)^{k-1}(1-\tilde{p})^{k-1}\nonumber\\
\hspace{-.3in}&&\qquad\times\big[p(\frac{k+1}{\lambda+\mu}+k\tilde{s}-\tilde{s}+\frac{1}{\mu}) \nonumber\\
\hspace{-.3in}&& +(1-p)\tilde{p}(\frac{k+1}{\lambda+\mu}+k\tilde{s}-\tilde{s}+\tilde{t})\big]\bigg\} \nonumber \\
\hspace{-.3in}&=&\sum_{k=1}^\infty\frac{\mu}{\lambda+\mu}(\frac{\lambda}{\lambda+\mu})^k\bigg\{\sum_{j=0}^{k-1}(1-p)^{j}(1-\tilde{p})^{j} \nonumber\\
\hspace{-.3in}&&\qquad\times\big[p(\frac{k+1}{\lambda+\mu}+j\tilde{s}+\frac{1}{\mu})\nonumber\\
\hspace{-.3in}&&+(1-p)\tilde{p}(\frac{k+1}{\lambda+\mu}+j\tilde{s}+\tilde{t})\big]\bigg\},  \label{nonpre_tt_eq}
\end{eqnarray}
\vspace{-.05in}
and 
\begin{eqnarray}
\hspace{-.3in} (1-\tilde{p})\tilde{s}&=&\frac{\mu}{\lambda+\mu}\frac{1}{\lambda+\mu}+\sum_{k=1}^\infty\frac{\mu}{\lambda+\mu}(\frac{\lambda}{\lambda+\mu})^k \nonumber\\
\hspace{-.3in}&&\qquad\times(1-p)^{k}(1-\tilde{p})^{k}(\frac{k+1}{\lambda+\mu}+k\tilde{s}). \label{nonpre_st_eq}
\end{eqnarray}
%
From \eqref{nonpre_pt_eq1}, we get: 
\begin{eqnarray}
\lambda(1-p)\tilde{p}^2+(\mu-\lambda+2\lambda p)\tilde{p}-\lambda p=0,
\label{nonpre_pt_eq2}
\end{eqnarray}
which leads to
\begin{eqnarray}
\tilde{p}=
\begin{cases}
\frac{-(\mu-\lambda+2\lambda p)+\sqrt{(\lambda+\mu)^2-4\lambda\mu(1-p)}}{2\lambda(1-p)},\  &0<p<1 \\
\frac{\lambda}{\lambda+\mu},\ &p=1
\end{cases}.
\label{nonpre_pt_res}
\end{eqnarray}
Solving \eqref{nonpre_tt_eq} and \eqref{nonpre_st_eq}, and using  \eqref{nonpre_pt_eq2}  give us
\begin{eqnarray}
\hspace{-.3in}(1-\tilde{p})\tilde{s}&=&\frac{1-\tilde{p}}{\lambda+\mu-2\lambda(1-p)(1-\tilde{p})},\label{st2}\\
\hspace{-.3in}\tilde{p}\tilde{t}&=&\frac{\lambda p+2\lambda p^2+(\lambda-2\lambda p^2-\mu+\mu p)\tilde{p}}{\mu p[\lambda+\mu-2\lambda(1-p)(1-\tilde{p})]}. \label{tt2}
\end{eqnarray}
\vspace{-.15in}
\subsection{Computing PAoI}
\vspace{-.05in}
Now we  compute PAoI shown in Fig. \ref{fig_nonpre}. Define $\pi(t)$ as the number of packets in the system (including the packet being served) at time $t$. So $\pi(t)=0$ means the system is free at time $t$. Different from the preemptive case, here $\pi(a(n_i))$ and $\pi(u(n_i))$ will respectively affect $W_{n_i}$ and $\hat{Z}_{n_i}$, in that they affect the degree to which new packets need to wait for service completion.  

We first compute the number of packets an arrival in $\Psi$ sees when it arrives. Since $\Psi$ is a special set of packets, they do not see exactly as what an ordinary packet will see. 
To this end, 
we define for each $k$
\begin{eqnarray}
p_k&\triangleq&\mathbb{P}[\pi(a(n))=k|n\in\Psi] \nonumber\\
&=&\frac{\mathbb{P}[\pi(a(n))=k,n\in\Psi]}{\mathbb{P}(n\in\Psi)}.
\end{eqnarray}
Consider the waiting time $W_n$ of packet $n$. If $\pi(a(n))=0$ then $W_n=0$. Otherwise $n$ needs to wait for the completion of the current service and the services of packets which arrive during the current service, till the server starts to serve a packet arriving no later that $a(n)$. Since the exponential distribution is memoryless, for $\pi(a(n))>0$, $W_n$ is the same as the process $\bar{S}_{\bar{n}}$ of a virtual packet $\bar{n}$ with $u(\bar{n})=a(n)$, and $n\in\Psi$ is equivalent to $(\hat{Z}_{\bar{n}}> \bar{S}_{\bar{n}})\cap (n\in \Phi)$. 
For a steady-state $M/M/1$ queue, we know that $\mathbb{P}[\pi(t)=k]=(1-\frac{\lambda}{\mu})(\frac{\lambda}{\mu})^k$. Thus, 
\begin{eqnarray*}
p_0&=&\frac{(1-\frac{\lambda}{\mu})p}{\mathbb{P}[n\in\Psi]}, \\
p_k&=&\frac{(1-\frac{\lambda}{\mu})(\frac{\lambda}{\mu})^k(1-\tilde{p})p}{\mathbb{P}[n\in\Psi]}, k\geq1. 
\end{eqnarray*}
Moreover, $\sum_{k=0}^\infty p_k$=1. Therefore, 
\begin{eqnarray*}
p_0&=&\frac{\mu-\lambda}{\mu-\lambda\tilde{p}}, \\
p_k&=&\frac{\mu-\lambda}{\mu-\lambda\tilde{p}}(1-\tilde{p})(\frac{\lambda}{\mu})^k, k\geq1.
\end{eqnarray*}
Hence, the waiting time can be computed as: 
\begin{eqnarray*}
\mathbb{E}\{W_{n_i}|n_i\in\Psi\}&=& \sum_{k=0}^\infty p_k\mathbb{E}\{W_{n_i}|\pi(a(n_i))=k,n_i\in\Psi\} \\
&=&p_0\cdot0+(1-p_0)\tilde{s} \\
&=&\frac{\lambda(1-\tilde{p})}{(\mu-\lambda\tilde{p})[\lambda+\mu-2\lambda(1-p)(1-\tilde{p})]}.
\vspace{-.1in}
\end{eqnarray*}
For $\mathbb{E}\{\hat{Z}_{n_i}|n_i\in\Psi\}$, define: 
\begin{eqnarray*}
& z_k\triangleq \mathbb{E}\{\hat{Z}_{n}|\pi(u(n))=k, n\in\Psi\}= \mathbb{E}\{\hat{Z}_{n}|\pi(u(n))=k\}.
\end{eqnarray*}
For $\mathbb{E}\{\hat{Z}_{n}|\pi(u(n))=k\}$, if a packet $n_j$ arrives during $S(n)$ (with probability $\frac{\lambda}{\lambda+\mu}$), it will wait $W_{n_j}=\bar{S}_{\bar{n}_j}$ before being served, with $\bar{n}_j$ a virtual packet defined as before. Since  $\pi(u(n_j))=k$, if $\hat{Z}_{\bar{n}_j}> \bar{S}_{\bar{n}_j}$ and $n_j\notin \Phi$, the (expected) remaining time of $\hat{Z}_{n}$ from $u(n_j)$ is still $z_k$. Otherwise no packets arrives during $S(n)$, giving us $\pi(d(n))=k-1$. Based on the above analysis, we get:
\begin{eqnarray}
\hspace{-.2in}&&z_1=\frac{\mu}{\lambda+\mu}\big[\frac{1}{\lambda+\mu}+\frac{1}{\lambda}+p\frac{1}{\mu}+(1-p)z_1\big] \label{T1}\\
\hspace{-.2in}&&\qquad\qquad +\frac{\lambda}{\lambda+\mu}\big[\frac{1}{\lambda+\mu}+\tilde{p}\tilde{t}+(1-\tilde{p})p(\tilde{s}+\frac{1}{\mu}) \nonumber \\
\hspace{-.2in}&&\qquad\qquad\qquad\qquad\qquad\qquad+(1-\tilde{p})(1-p)(\tilde{s}+z_1)\big],  \nonumber 
\end{eqnarray}
and that for general $k$, 
\begin{eqnarray}
\hspace{-.2in}&&z_k=\frac{\mu}{\lambda+\mu}(\frac{1}{\lambda+\mu}+z_{k-1}) \label{Tk} \\
\hspace{-.2in}&&\qquad\qquad+\frac{\lambda}{\lambda+\mu}\big[\frac{1}{\lambda+\mu}+\tilde{p}\tilde{t}+(1-\tilde{p})p(\tilde{s}+\frac{1}{\mu}) \nonumber\\
\hspace{-.2in}&&\qquad\qquad\qquad\qquad\qquad\qquad+(1-\tilde{p})(1-p)(\tilde{s}+z_k)\big]. \nonumber
\end{eqnarray}
 Solving \eqref{T1} gives us
\begin{eqnarray}
z_1=\frac{\mu+\lambda+\lambda p+\lambda^2\tau}{\lambda[\lambda+\mu p-\lambda(1-p)(1-\tilde{p})]}, \label{T1_res}
\end{eqnarray}
where
\begin{eqnarray}
\tau=\frac{(\lambda+\mu)p+(\lambda+\mu)p^2+[\lambda+(\mu-\lambda)p^2-\mu]\tilde{p}}{\mu p[\lambda+\mu-2\lambda(1-p)(1-\tilde{p})]}.  \label{alpha}
\end{eqnarray}
From \eqref{Tk}, we can get
\begin{eqnarray*}
&&z_k-\frac{(1-\tilde{p})(1+\lambda\tau)}{\mu\tilde{p}}\\
&&\qquad\qquad\qquad =(1-\tilde{p})[z_{k-1}-\frac{(1-\tilde{p})(1+\lambda\tau)}{\mu\tilde{p}}].
\end{eqnarray*}
The evolution of the LCFS queueing  system shows that if a packet $n$ sees no more than one packet when it arrives, then there will be only one packet in the system (packet $n$ itself) when it starts to receive service. Thus, 
\begin{align}
\mathbb{P}[\pi(u(n))=1|n\in\Psi] &=p_0+p_1, \\
\mathbb{P}[\pi(u(n))=k|n\in\Psi] &=p_k,k\geq2.
\end{align}
Hence, 
\begin{eqnarray*}
\hspace{-.3in}&&\mathbb{E}\{\hat{Z}_{n_i}|n_i\in\Psi\} \\
\hspace{-.3in}&&= \sum_{k=0}^\infty \mathbb{P}[\pi(u(n_i))=k|n_i\in\Psi]\mathbb{E}\{\hat{Z}_{n_i}|\pi(u(n_i))=k,n_i\in\Psi\}\\
\hspace{-.3in}&&=(p_0+p_1)z_1+\sum_{k=2}^\infty p_kz_k \\
\hspace{-.3in}&& = p_0z_1+\sum_{k=1}^\infty p_kz_k \\
\hspace{-.3in}&&=p_0z_1+\sum_{k=1}^\infty p_k[z_k-\frac{(1-\tilde{p})(1+\lambda\tau)}{\mu\tilde{p}}] \\
\hspace{-.3in}&&\qquad\qquad+\sum_{k=1}^\infty p_k\frac{(1-\tilde{p})(1+\lambda\tau)}{\mu\tilde{p}} \\
\hspace{-.3in}&& = p_0z_1+\sum_{k=1}^\infty p_k(1-\tilde{p})^{k-1}[z_1-\frac{(1-\tilde{p})(1+\lambda\tau)}{\mu\tilde{p}}] \\
\hspace{-.3in}&&\qquad\qquad+\sum_{k=1}^\infty p_k\frac{(1-\tilde{p})(1+\lambda\tau)}{\mu\tilde{p}} \\
\hspace{-.3in}&&=\frac{\mu(\mu-\lambda)(\mu+\lambda+\lambda p+\lambda^2\tau)}{\lambda(\mu-\lambda\tilde{p})[\mu-\lambda(1-\tilde{p})][\lambda+\mu p-\lambda(1-p)(1-\tilde{p})]} \\
\hspace{-.3in}&&\qquad\qquad\qquad+\frac{\lambda^2(1-\tilde{p})^2(1+\lambda\tau)}{\mu(\mu-\lambda\tilde{p})[\mu-\lambda(1-\tilde{p})]}. 
\end{eqnarray*}
Therefore, PAoI can be computed as: 
\begin{eqnarray}
\hspace{-.4in}&&A_P^{LCFS,non}=  \frac{\lambda(1-\tilde{p})}{(\mu-\lambda\tilde{p})[\lambda+\mu-2\lambda(1-p)(1-\tilde{p})]} \nonumber\\
\hspace{-.4in}&&+\frac{\mu(\mu-\lambda)(\mu+\lambda+\lambda p+\lambda^2\tau)}{\lambda(\mu-\lambda\tilde{p})[\mu-\lambda(1-\tilde{p})][\lambda+\mu p-\lambda(1-p)(1-\tilde{p})]} \nonumber\\
\hspace{-.4in}&&+\frac{\lambda^2(1-\tilde{p})^2(1+\lambda\tau)}{\mu(\mu-\lambda\tilde{p})[\mu-\lambda(1-\tilde{p})]}, \label{paoi_nonpre}
\end{eqnarray}
where $\tau$ is given in \eqref{alpha} and $\tilde{p}$ is given in \eqref{nonpre_pt_res}.
From this result, we can see that even in the case of $p=1$, the solution is non-trivial. 


\vspace{-.05in}
\section{PAoI under Retransmission with Preemptive Priority}\label{section:retran-preemptive}
\vspace{-.05in}

In this case, we consider the case when a packet is transmitted repeatedly until it reaches the destination successfully or it is preempted. Thus, $\Phi=\Psi$. Here the ``departure" of a packet means the moment it is transmitted successfully or preempted. Actually, since we do not assume feedback, a packet will still be served before the arrival of the next packet even if it has been transmitted successfully, but that has no influence to the system due to the preemptive priority. We can regard this policy as only storing the latest packet and replace it with a new one as soon as a new packet arrives.
\vspace{-.1in} 
\subsection{Analyzing a Service Process}

We divide PAoI in the same way as in Section \ref{section:preemptive} and again re-define the following symbols: 
\begin{eqnarray*}
\tilde{p}&\triangleq& \mathbb{P}(n\in \Psi), \\
\tilde{t}&\triangleq& \mathbb{E}\{\hat{S}_n|n\in \Psi\}, \\
\tilde{s}&\triangleq& \mathbb{E}\{\hat{S}_n\},
\end{eqnarray*}
i.e., $\tilde{p}$ is the probability that there is no packets arriving during $\hat{S}_n$, or $\hat{S}_n=S_n$. Other than the total service time $S_n$, we further use $S_{n,k}$ to denote the $k$-th service of packet $n$.

For the $\hat{S}_n$ process, in $S_{n,1}$, if $n+1$ arrives (with probability $\frac{\lambda}{\lambda+\mu}$), it will preempt $n$ and the remaining time of $\hat{S}_n$ from $a(n+1)$ is $\hat{S}_{n+1}$. Otherwise if $n$ gets lost upon service completion, the server will start to retransmit $n$ and $\mathbb{E}\{\hat{S}_n-S_{n,1}|X_n>S_{n,1}, \hat{S}_n>S_{n,1}\}=\tilde{s}$, since the first transmission does not influence the following retransmissions.
Based on this observation, we get: 
\begin{eqnarray}
\tilde{s}&=&\frac{\mu}{\lambda+\mu}\big[p\frac{1}{\lambda+\mu}+(1-p)(\frac{1}{\lambda+\mu}+\tilde{s})\big]\nonumber\\
&&\qquad+\frac{\lambda}{\lambda+\mu}(\frac{1}{\lambda+\mu}+\tilde{s}),
\end{eqnarray}
which leads to:
\begin{eqnarray}
\tilde{s}=\frac{1}{p\mu}.
\end{eqnarray}
If $n\in \Psi$, we know $n$ has been transmitted successfully before it is preempted. By considering the number of transmissions it takes to successfully transmit $n$ and using Lemma \ref{lem}, we have:
\begin{eqnarray}
\tilde{p}&=&\sum_{k=0}^\infty(1-p)^k p (\frac{\mu}{\lambda+\mu})^{k+1}, \label{retran_preem_ph}\\
\tilde{p}\tilde{t}&=&\sum_{k=0}^\infty(1-p)^k p (\frac{\mu}{\lambda+\mu})^{k+1}\frac{k+1}{\lambda+\mu}. \label{retran_preem_th}
\end{eqnarray}
By solving \eqref{retran_preem_ph} and \eqref{retran_preem_th}, we can get
\begin{eqnarray}
\tilde{p}&=&\frac{p\mu}{\lambda+p\mu}, \\
\tilde{t}&=&\frac{1}{\lambda+p\mu}.
\end{eqnarray}

\subsection{Computing PAoI}
Now we are going to compute the PAoI. 
%
For $\hat{Y}_{n_i}$, the expected time from $d(n_i)$ to $a(n_{i}+1)$ is $\frac{1}{\lambda}$ and the process from $a(n_{i}+1)$ to $d(n_{i+1})$  is $\hat{S}_{n_i+1}$. Hence
\begin{eqnarray*}
\mathbb{E}\{\hat{Y}_{n_i}|n_i\in \Psi\}=\frac{1}{\lambda}+\tilde{s}.
\end{eqnarray*}
Thus, we can compute PAoI as
\begin{eqnarray}
A_P^{RT,pre}&=& \mathbb{E}\{S_{n_i}+\hat{Y}_{n_i}|n_i\in\Psi\} \nonumber \\
&=& \mathbb{E}\{\hat{S}_{n_i}|n_i\in\Psi\}+\mathbb{E}\{\hat{Y}_{n_i}|n_i\in\Psi\}  \nonumber\\
&=&\tilde{t}+\frac{1}{\lambda}+\tilde{s} \nonumber \\
&=&\frac{1}{\lambda+p\mu}+\frac{1}{\lambda}+\frac{1}{p\mu}. \label{paoi_retran_preem}
\end{eqnarray}

\emph{Remark:} 
It turns out that this result corresponds to the result under the LCFS with preemptive priority policy with a service rate $p\mu$ and a success probability $1$. 
%
This is intuitive since in the LCFS with preemptive priority case with $p=1$, each packet is either transmitted successfully or preempted, while in this case each packet is still either transmitted successfully or preempted, with a mean service time $\frac{1}{p\mu}$.

\section{PAoI under Retransmission with Non-Preemptive Priority}\label{section:retran-nonpreemptive}

Under the Retransmission with non-preemptive priority policy, the server keeps retransmitting the most recent packet no matter it has been successfully transmitted or not. In this case, the most recent packet is kept in the queue and the server is always busy. Since a new arrival can't interrupt the current service, a packet may be replaced by a more recent packet when it is waiting in the queue or upon a service completion. Here the ``departure" of a packet means the moment it is transmitted successfully or replaced. Let $\Omega$ denote the set of packets which are not replaced while waiting in the queue, i.e., the ones that have been served before departure. Note that in this case $\Phi=\Psi$ and $\Phi \subset \Omega$.

\subsection{Analyzing a Service Process}
We divide PAoI in the same way as in Section \ref{section:nonpreemptive}, 
but here $W_n$, $S_n$, $\hat{S}_n$ and $\hat{Z}_n$ are only meaningful for packet $n\in \Omega$, which means that there is no packets arriving during $W_n$.
%
Again we define the following notions for packet $n\in \Omega$:
\begin{eqnarray*}
\tilde{p}&\triangleq& \mathbb{P}(n\in \Psi), \\
\tilde{t}&\triangleq& \mathbb{E}\{\hat{S}_n|n\in \Psi\}, \\
\tilde{s}&\triangleq& \mathbb{E}\{\hat{S}_n\},
\end{eqnarray*}
i.e., $\tilde{p}$ is the probability that $n$ has been transmitted successfully before it is replaced by a more recent packet, or $\hat{S}_n=S_n$.

For the $\hat{S}_n$ process, if n reaches the destination successfully after $S_{n,1}$, then $\hat{S}_n$ ends. Otherwise the server will start to transmit another packet $\hat{n}\in \Omega$ which arrives during $S_{n,1}$, or retransmit $n$, both resulting in the expected remaining time of $\hat{S}_n$ as $\tilde{s}$. 
This observation gives us
\begin{eqnarray}
\tilde{s}=p\frac{1}{\mu}+(1-p)(\frac{1}{\mu}+\tilde{s}),
\end{eqnarray}
which leads to:
\begin{eqnarray}
\tilde{s}=\frac{1}{p\mu}.
\end{eqnarray}
If $n\in \Psi$, we know $n$ has been transmitted successfully before it is replaced. By considering the number of transmissions it takes to successfully transmit $n$ and using Lemma \ref{lem}, we have:
\begin{eqnarray}
\tilde{p}&=&\sum_{k=0}^\infty(1-p)^k(\frac{\mu}{\lambda+\mu})^{k} p, \label{retran_nonpre_ph}\\
\tilde{p}\tilde{t}&=&\sum_{k=0}^\infty(1-p)^k(\frac{\mu}{\lambda+\mu})^{k} p(\frac{k}{\lambda+\mu}+\frac{1}{\mu}). \label{retran_nonpre_th}
\end{eqnarray}
By solving \eqref{retran_nonpre_ph} and \eqref{retran_nonpre_th}, we can get
\begin{eqnarray}
\tilde{p}&=&\frac{p(\lambda+\mu)}{\lambda+p\mu}, \\
\tilde{t}&=&\frac{1}{\mu}+\frac{(1-p)\mu}{(\lambda+\mu)(\lambda+p\mu)}.
\end{eqnarray}

\subsection{Computing PAoI}
Now we are going to compute the PAoI. 
Remember the server is always busy and $n_i\in \Psi$ indicates that there is no arrivals during $n_i$'s waiting time. Thus, 
\begin{eqnarray*}
\mathbb{E}\{W_{n_i}|n_i\in\Psi\}=\frac{1}{\lambda+\mu}. 
\end{eqnarray*}
%

Consider the period $\hat{Z}_{n_i}$. We have $\mathbb{E}\{S_{n_i}|n_i\in\Psi\}=\mathbb{E}\{\hat{S}_{n_i}|n_i\in\Psi\}=\tilde{t}$. If there is a packet $\hat{n_i}$ with $a(\hat{n_i})>a(n_i)$ in the queue at $d(n_i)$ ($\hat{n_i}$ can only arrives during the last service of $n_i$ before $d(n_i)$, and the probability is $\frac{\lambda}{\lambda+\mu}$), the process from $d(n_i)$ to $d(n_{i+1})$ is $\hat{S}_{\hat{n_i}}$. Otherwise,  
the expected time from $d(n_i)$ to $a(n_{i}+1)$ is $\frac{1}{\lambda}$. After that, because the server is always busy, it still needs to complete a service of $n_i$ before it starts to transmit a packet arrives after $n_i$. Based on these observations, we have:
\begin{eqnarray*}
\mathbb{E}\{\hat{Z}_{n_i}|n_i\in\Psi\}=\tilde{t}+\frac{\lambda}{\lambda+\mu}\tilde{s}+\frac{\mu}{\lambda+\mu}(\frac{1}{\lambda}+\frac{1}{\mu}+\tilde{s}). 
\end{eqnarray*}
So PAoI can be computed as:
\begin{eqnarray}
A_P^{RT,non}&=&\frac{1}{\lambda+\mu}+\tilde{t}+\frac{\lambda}{\lambda+\mu}\tilde{s}+\frac{\mu}{\lambda+\mu}(\frac{1}{\lambda}+\frac{1}{\mu}+\tilde{s}) \nonumber \\
&=&\frac{1}{\mu}+\frac{1}{\lambda+p\mu}+\frac{1}{\lambda}+\frac{1}{p\mu}.\label{paoi_retran_nonpre}
\end{eqnarray}

\section{Numerical results}\label{section:numerical}
We present numerical evaluations of PAoI under different scheduling policies, including FCFS, FCFS with packet management (the $M/M/1/2^*$ scheme in \cite{costa}), LCFS with preemptive priority, LCFS with non-preemptive priority, Retransmission with preemptive priority and Retransmission with non-preemptive priority. Note that the $M/M/1/2^*$ scheme in \cite{costa} is equivalent to the LCFS with non-preemptive priority policy that discards all stale packets.
The service rate is set to $\mu=1$ while the arrival rate is varied to show performances under different channel utilizations $\rho=\frac{\lambda}{\mu}$. The cases $p=0.1$, $p=0.5$ and $p=1$ are selected to represent different delivery error regimes. We present not only the results computed from our formulas \eqref{paoi_fcfs}, \eqref{paoi_preem}, \eqref{paoi_nonpre}, \eqref{paoi_retran_preem} and \eqref{paoi_retran_nonpre}, but also those obtained by simulating real queueing systems with the corresponding settings.
\begin{figure}[!t]
\centering
\includegraphics[width=3.6in]{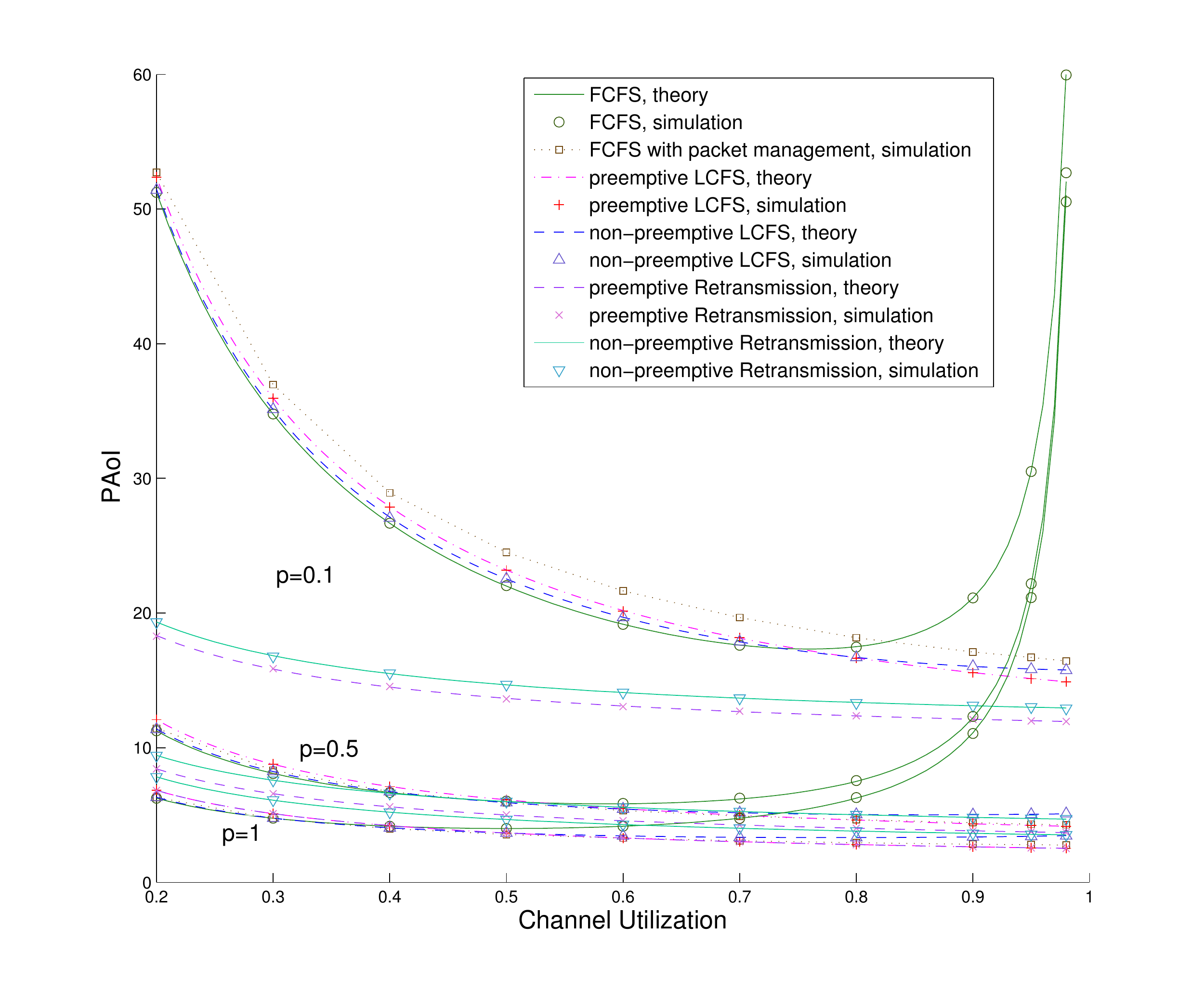}
\vspace{-.1in}
\caption{PAoI in different queueing  systems with packet loss.}
\label{fig_sim}
\vspace{-.1in}
\end{figure}

From Fig. \ref{fig_sim}, we see that the simulation results  match our theoretical results very well. 
We can see that when channel utilization is high, the PAoI under FCFS becomes very large due to large queueing delay, while other policies effectively avoid this problem. 
On the other hand, when packet loss rate is high, FCFS with packet management suffers from the lack of packet deliveries but LCFS again ensures a low PAoI, matching our intuition about the benefits of LCFS.
Moreover, retransmission policies have significant reductions on PAoI compared to other policies when packet loss rate is high. But when packet loss rate is low, Retransmission with non-preemptive priority suffers a performance loss since retransmissions can also block later packets.

\section{Conclusion}\label{section:conclusion}
 We consider the peak age-of-information (PAoI) in an $M/M/1$ queueing  system with packet delivery failure, a setting that  models  real-world situations with transmission errors. 
We derive exact PAoI expressions under different scheduling policies, including FCFS, LCFS with preemptive priority, LCFS with non-preemptive priority, Retransmission with preemptive priority, and Retransmission with non-preemptive priority. Our analytical and simulation results show that the LCFS principle as well as retransmissions can successfully avoid increments in PAoI resulting from large queueing delay and packet loss.


\section*{Acknowledgment}

The authors would like to thank Prof. Eytan Modiano at MIT for the motivating discussions and valuable comments.  

This work  was supported in part by the National Basic Research Program of China Grant 2011CBA00300, 2011CBA00301, the National Natural Science Foundation of China Grant 61361136003, 61303195, Tsinghua Initiative Research Grant, Microsoft Research Asia Collaborative Research Award, and the China Youth 1000-talent Grant.




\begin{thebibliography}{1}

\bibitem{corke} P. Corke, T. Wark, R. Jurdak, et al. Environmental wireless sensor networks. \emph{Proceedings of the IEEE}, 2010, 98(11): 1903-1917.

\bibitem{papa} P. Papadimitratos, A. La Fortelle, K. Evenssen, et al. Vehicular communication systems: Enabling technologies, applications, and future outlook on intelligent transportation. \emph{Communications Magazine}, IEEE, 2009, 47(11): 84-95.

\bibitem{reddy} A. Reddy, S. Banerjee, A. Gopalan, et al. On distributed scheduling with heterogeneously delayed network-state information. \emph{Queueing Systems}, 2012, 72(3-4): 193-218.

\bibitem{kaul} S. Kaul, R. Yates, and M. Gruteser. Real-time status: How often should one update? \emph{Proceedings of IEEE International Conference on Computer Communications, INFOCOM}, 2012.

\bibitem{yates} R. D. Yates and S. Kaul. Real-time status updating: Multiple sources. \emph{Proceedings of IEEE ISIT}, 2012.

\bibitem{kaul2} S. K. Kaul, R. D. Yates, and M. Gruteser. Status Updates Through Queues. \emph{Proceedings of IEEE Annual Conference on Information Sciences and Systems, CISS}, 2012.

\bibitem{kam} C. Kam, S. Kompella, and A. Ephremides. Age of information under random updates. \emph{Proceedings of IEEE ISIT}, 2013.

\bibitem{costa} M. Costa, M. Codreanu, and A. Ephremides. Age of information with packet management. \emph{Proceedings of IEEE ISIT}, 2014.

\bibitem{huang} L. Huang and E. Modiano. Optimizing Age-of-Information in a Multi-class Queueing System. \emph{Proceedings of IEEE ISIT}, 2015.

\bibitem{guo} X. Guo, R. Singh, P.R. Kumar, et al. A High Reliability Asymptotic Approach for Packet Inter-Delivery Time Optimization in Cyber-Physical Systems. \emph{Proceedings of ACM MobiHoc}, 2015.

\bibitem{sun} Y. Sun, E. Uysal-Biyikoglu, R. Yates, et al. Update or Wait: How to Keep Your Data Fresh. \emph{arXiv preprint arXiv:1601.02284}, 2016.

\end{thebibliography}
%

\end{document}